\documentclass[a4paper, 10pt]{article}

\usepackage{amsfonts, amssymb}
\usepackage{amscd}
\usepackage{amsmath}
\usepackage{comment}
\usepackage{pdflscape}
\usepackage{amsthm}
\usepackage{color}
\usepackage{tabularx}
\usepackage{xfrac}
\usepackage[table]{xcolor}
\usepackage{enumitem}
\usepackage{marvosym}
\usepackage{amsmath}

\newtheorem{theo}{Theorem}[section]
\newtheorem{defi}[theo]{Definition}
\newtheorem{coro}[theo]{Corollary}

\newtheorem{lemm}[theo]{Lemma}

\newtheorem{ex}{Example}

\newcommand{\Z}{{\mathbb{Z}}}

\newcommand{\C}{{\cal C}}

\newcommand{\D}{{\cal D}}

\newcommand{\type}{{\left(\alpha,\beta,\theta;k_{0};k_{1},k_{2};k_{3},k_{4},k_{5}\right)}}

\newcommand{\T}{{\cal R}}
\newcommand{\R}{{\cal R_{\alpha,\beta,\theta}}}
\newcommand{\M}{{\Z_{2}\Z_{4}\Z_{8}}}

\title{$\Z_{2}\Z_{4}\Z_{8}$-Cyclic Codes}

\author{ Ismail Aydogdu$^{a}$\thanks{{\footnotesize iaydogdu@yildiz.edu.tr
(Ismail Aydogdu)}}, Fatmanur Gursoy$^{a}$\thanks{{\footnotesize fatmanur@yildiz.edu.tr (Fatmanur Gursoy)}}
 \\Department of Mathematics, Yildiz Technical University \\ Istanbul, Turkey }

\date{\today}

\begin{document}

\maketitle

\begin{abstract}
In this paper we study $\M$-additive codes, which are the extension of recently introduced $\mathbb{Z}_{2}\mathbb{Z}_{4}$-additive codes. We determine the standard forms of the generator and parity-check matrices of $\M$-additive codes. Moreover, we investigate $\M$-cyclic codes giving their generator polynomials and spanning sets. We also give some illustrative examples of both $\M$-additive codes and $\M$-cyclic codes.   

\end{abstract}

\begin{quotation}
\end{quotation}

\section{Introduction}
In the milestone paper written by Sloane and co-workers(\cite{Hammons}) in 1994, it was shown that interesting binary codes could be found as images of linear codes over $\Z_{4}$ under a non-linear Gray map. This paper have attracted researchers to study codes over different rings and nowadays there has been many papers on codes over rings. The reader may see some of them in \cite{Calder,Schmidt,Honold}. There are many reasons for study codes over rings. First of all, the properties of rings are very close to the properties of finite fields.  Hence a theory of linear codes over finite chain rings is expected to resemble  the  theory  of  linear  codes  over  finite fields. Moreover,  the  class  of  finite chain  rings  contains  important  representatives  like  integer  residue  rings  of  prime power order and Galois rings.

Later in 2010, Borges et al. brought a new perspective to codes over rings, introducing $\Z_2\Z_4$-additive codes \cite{Borges}. $\Z_2\Z_4$-additive codes are $\Z_4$ submodules (additive groups) of $\Z_2^{\alpha}\times\Z_4^{\beta}$ where $\alpha$ and $\beta$ are positive integers. If $\alpha=0$ then $\Z_2\Z_4$-additive codes are quaternary linear codes over $\Z_{4}$ and if $\beta=0$ then $\Z_2\Z_4$-additive codes are just binary linear codes. $\Z_2\Z_4$-additive codes also have some applications in fields such as the field of Steganography \cite{Steo}. Recently, Aydogdu and Siap generalized these additive codes to codes over $\Z_2\times\Z_{2^s}$ (\cite{ismail1}) and $\Z_{p^r}\times\Z_{p^s}$ (\cite{ismail2}) for a prime $p$, $r$ and $s$ are positive integers with $1\leq r<s$. Another considerable work about $\Z_2\Z_4$-additive codes have been done by Abualrub et al. in 2014, in which they have introduced $\Z_2\Z_4$-cyclic codes and have given the generator polynomials for these cyclic codes \cite{Abu}. Lately, Borges et al. study the duals of $\Z_2\Z_4$-cyclic codes in \cite{Borges2}.

In this work, we aim to study the structure of $\M$-additive and cyclic codes. The reader may see such an additive codes as an extension or generalization of $\Z_{2}\Z_{4}$-additive codes. We begin with the structure of $\M$-additive codes and give the standard forms of generator and parity-check matrices. We also relate these codes to binary codes by defining a special Gray map. Next, we determine the generator polynomials and minimal spanning sets of $\M$-cyclic codes. Furthermore, we present some examples both the additive and cyclic $\M$ codes.


\section{$\mathbb{Z}_2\mathbb{Z}_4\Z_{8}$-additive codes}

 Let $\mathbb{Z}_{2}$ be the finite binary field, $\mathbb{Z}_{4}$ and $\mathbb{Z}_{8}$ be the ring of integers modulo $4$ and modulo $8$ respectively. We construct the following set

\begin{equation*}
\T=\mathbb{Z}_{2}\Z_{4}\Z_{8}=\left\{ \left( u,v,w\right) |u\in\mathbb{Z}_{2},~v\in\Z_{4}\text{ and }w\in \Z_{8}\right\}.
\end{equation*}
It is obvious that this set is closed under addition. To make it closed under multiplication by elements in $\Z_{8}$ we define the following multiplication for $\left( u,v,w\right)\in \T$ and $d\in \Z_{8}$.

\begin{equation*}
d \cdot \left( u,v,w\right) =\left( du~mod~2,dv~mod~4,dw\right).
\end{equation*}
This discussion shows that the set $\T$ is also closed under multiplication by elements in $\Z_{8}$ and therefore $\T$ is a $\Z_{8}$-module with respect to this scalar multiplication.

\begin{defi}
$\C$ is called a $\mathbb{Z}_{2}\mathbb{Z}_{4}\Z_{8}$-additive code if it is a subgroup of $\mathbb{Z}_{2}^{\alpha}\times \Z_{4}^{\beta}\times \Z_{8}^{\theta}$.
\end{defi}

It is clear from the definition of a $\M$-additive code that the first $\alpha$ coordinates of $\C$ consist of the entries from $\Z_{2}$, the next $\beta$ coordinates are elements from $\Z_{4}$ and remanning $\theta$ coordinates are the elements of the ring $\Z_{8}$. We all know very well from the Finite Abelian Group Theorem that such a additive code $\C$ which is a subgroup of $\mathbb{Z}_{2}^{\alpha}\times \Z_{4}^{\beta}\times \Z_{8}^{\theta}$ is group isomorphic to the abelian structure $$\Z_{2}^{k_{0}}\times\Z_{2}^{2k_{1}}\times\Z_{2}^{k_{2}}\times\Z_{2}^{3k_{3}}\times\Z_{2}^{2k_{4}}\times\Z_{2}^{k_{5}}.$$ 
Considering this isomorphism, we say such a $\M$-additive code $\C$ is of type $\type$. 

In the literature, it has been defined several Gray maps to relate codes over rings to codes over $\Z_{2}$(binary codes) with respect to different metrics. For instance, Carlet has defined a Gray map for codes over $\Z_{2^{s}}$ with respect to the homogenous weight in \cite{Carlet}. Using this result we give the following definition of a generalized Gray map.

\begin{defi}
Let $\phi_{1}$ and $\phi_{2}$ are the following well-known Gray maps.
\begin{equation*}
\begin{split}
\phi_{1}:&\Z_{4}\rightarrow\Z_{2}^{2}  \\
                &0 \rightarrow 00\\
                &1 \rightarrow 01  \\
                &2 \rightarrow 11    \\
                &3 \rightarrow 10\\\\\\\\\\
\end{split}
\qquad\qquad\qquad
\begin{split}
 \phi_{2}: &\Z_{8}\rightarrow \Z_{2}^{4}\\
  &0 \rightarrow 0000\\
&1 \rightarrow 0001 \\
&2 \rightarrow 0011\\
&3 \rightarrow 0111\\
&4 \rightarrow 1111  \\
&5 \rightarrow 1110\\
&6 \rightarrow 1100\\
&7 \rightarrow 1000
\end{split}
\end{equation*}
We can also define a Gray map for codes over $\M$ for all $u=(u_{0},u_{1},...,u_{\alpha-1})\in\mathbb{Z}_{2}^{\alpha}$,  $v=(v_{0},v_{1},...,v_{\beta-1})\in\mathbb{Z}_{4}^{\beta}$ and $w=(w_{0},w_{1},...,w_{\theta-1})\in \Z_{8}^{\theta}$ as follows.

\begin{eqnarray*}
\Phi&:& \Z_{2}^{\alpha}\times\Z_{4}^{\beta} \times \Z_{8}^{\theta}\rightarrow\Z_{2}^n \\
\Phi (u,v,w)&=&\left( u_{0},u_{1},...,u_{\alpha-1},\phi_{1}(v_{0}),\phi_{1}(v_{1}),...,\phi_{1}(v_{\beta-1}),\phi_{2} (w_{0}),\phi_{2}(w_{1}),...,\phi_{2} (w_{\theta-1})\right).
\end{eqnarray*}
Hence, the Gray image $\Phi\left(\C\right)=C$ of a $\Z_{2}\Z_{4}\Z_{8}$-additive code $\C$ is a binary code of length $n=\alpha+2\beta+4\theta$ and called $\M$-linear code.
\end{defi}

\subsection{Generator matrices of $\Z_{2}\Z_{4}\Z_{8}$-additive codes}

The generator matrix $G$ of a linear code is the matrix with rows that are formed
by a set of basis elements (minimal spanning set) of the linear subspace $C$ in case of the finite
fields or the subgroup (or submodule) in a more general setting. The code $C$ is constituted by the all linear combinations of
the rows of $G$. Here, we determine the standard form of the generator matrices of $\Z_{2}\Z_{4}\Z_{8}$-additive codes. The standard form of the matrix $G$ is a special form of the matrix which is obtained by elementary row operations. Using the standard form of the generator matrix we can easily determine the type of a code and then calculate its size directly.

\begin{theo}
Let $\C$ be a $\M$-additive code of type $\type$. Then $\C$ is permutation equivalent to a $\M$-additive code which has the following standard form generator matrix
\begin{eqnarray}\label{G}
G_{S} =
 \left(\begin{array}{cc|ccc|cccc}
  I_{k_0} & \bar{A}_{01} & 0 & 0 & 2T_{1} & 0 & 0& 0& 4T_{2}  \\
  0 & \bar{S_{1}} & I_{k_1} & B_{01} & B_{02}& 0 & 0& 2T_{3}& 2T_{4} \\
  0 & 0 & 0 & 2I_{k_2} & 2B_{12} & 0& 0& 0& 4T_{5}  \\
   0 & \bar{S_{2}}& 0 & S_{01} & S_{02} & I_{k_{3}}& A_{01}&A_{02}&A_{03}\\
0 & \bar{S_{3}}& 0 & 0 & 2S_{12} & 0& 2I_{k_{4}}&2A_{12}&2A_{13}\\
0 & 0& 0 & 0 & 0 & 0& 0 &4I_{k_{5}}&4A_{23}
 \end{array}\right)
\end{eqnarray}
where $\bar{A}_{01}, \bar{S_{1}}, \bar{S_{2}}, \bar{S_{3}}$ are matrices with all entries from $\Z_{2}$ and $B_{01}, B_{02} B_{12}, S_{01}, S_{02}, S_{12}, T_{1}$ are matrices over $\Z_{4}$. Also, $A_{ij}$ and $T_{k}$ are matrices with all entries from $\Z_{8}$ for $0\leq i< j\leq 3$ and  $2\leq k \leq5$. Furthermore, $\C$ has $2^{k_{0}}2^{2k_{1}}2^{k_{2}}2^{3k_{3}}2^{2k_{4}}2^{k_{5}}$ codewords.
\end{theo}

\begin{proof}
Let  $\C$ be a $\M$-additive code of length $\alpha+\beta+\theta$ where $\alpha$ is the length of $\Z_{2}$ part and $\beta$ and $\theta$ are lengths of $\Z_{4}$ and $\Z_{8}$ parts respectively. Therefore we can write the generator matrix of $\C$ in the following form.

\begin{eqnarray*}
 \left(\begin{array}{cc|ccc|cccc}
 \cellcolor{purple!50}I_{k_0} &\cellcolor{purple!50} \bar{A'_{01}} & T'_{01} & T'_{02} & T'_{03} & T'_{04} & T'_{05}& T'_{06}& T'_{07}  \\[6pt] \hline
  S'_{01} & S'_{02} & \cellcolor{cyan!50}I_{k_1} & \cellcolor{cyan!50}B'_{01} & \cellcolor{cyan!50}B'_{02}& T'_{14} & T'_{15}& T'_{16}& T'_{17} \\
  S'_{11} & S'_{12} & \cellcolor{cyan!50}0 & \cellcolor{cyan!50}2I_{k_2} & \cellcolor{cyan!50}2B'_{12} & T'_{24} & T'_{25}& T'_{26}& T'_{27} \\[6pt] \hline
S'_{21} & S'_{22}& S'_{23} & S'_{24} & S'_{25} & \cellcolor{green!50}I_{k_{3}}& \cellcolor{green!50}A'_{01}&\cellcolor{green!50}A'_{02}&\cellcolor{green!50}A'_{03}\\
S'_{31} & S'_{32}& S'_{33} & S'_{34} & S'_{35} &\cellcolor{green!50} 0& \cellcolor{green!50}2I_{k_{4}}&\cellcolor{green!50}2A'_{12}&\cellcolor{green!50}2A'_{13}\\
S'_{41} & S'_{42}& S'_{43} & S'_{44} & S'_{45} &\cellcolor{green!50} 0& \cellcolor{green!50}0 &\cellcolor{green!50}4I_{k_{5}}&\cellcolor{green!50}4A'_{23}
 \end{array}\right)
\end{eqnarray*}

By applying the necessary row operations to the above matrix, we get the standard form generator matrix in (\ref{G}).
\end{proof}

\begin{ex}
Let $\C$ be a $\M$-additive code generated by the matrix
$$G=\left(
\begin{array}{cc|ccc|cccc}
1 & 1 & 2 & 0 & 2 & 4 & 0 & 4 & 4 \\
0 & 1 & 3 & 2 & 1 & 6 & 2 & 6 & 4 \\
1 & 1 & 2 & 2 & 0 & 4 & 4 & 0 & 4 \\
1 & 0 & 1 & 3 & 2 & 5 & 7 & 2 & 6 \\
1 & 1 & 2 & 2 & 0 & 2 & 4 & 6 & 2\\
1 & 0 & 0 & 0 & 2 & 4 & 4 & 0 & 4
\end{array}\right).$$
This matrix can be written in the standard form by using elementary row operations as follows.
\begin{eqnarray}\label{G_s}
G_S=\left(
\begin{array}{cc|ccc|cccc}
\cellcolor{purple!50}1 & \cellcolor{purple!50}0 & 0 & 0 & 2 & 0 & 0 & 0 & 4 \\
\cellcolor{purple!50}0 & \cellcolor{purple!50}1 & 0 & 0 & 2 & 0 & 0 & 4 & 4 \\ [6pt] \hline
0 & 0 & \cellcolor{cyan!50}1 & 0 & 3 & 0 & 0 & 6 & 4 \\
0 & 0 & 0 & \cellcolor{cyan!50}2 & 0 & 0 & 0 & 4 & 0 \\ [6pt] \hline
0 & 0 & 0 & 1 & 1 & \cellcolor{green!50}1 & 1 & 2 & 0\\
0 & 0 & 0 & 0 & 0 & 0 & \cellcolor{green!50}2 & 6 & 6
\end{array}\right)
\end{eqnarray}
So, the matrix $G_S$ say that $\C$ is of type $\left(2,3,4;2;1,1;1,1,0 \right)$ and  $\C$ has $2^{2}\cdot 4^{1}\cdot2^{1}\cdot8^{1}\cdot2^{2}\cdot4^{1}=4096$ codewords.
\end{ex}

\subsection{Parity-check matrices of $\M$-additive codes}

The set of all vectors which are orthogonal to every vector in $C$ is a subspace, and hence a linear code called the dual code of
$C$, and denoted by $C^{\perp}$. A generator matrix for $C^{\perp}$ is called a parity-check matrix of $C$. In this subsection, we determine the standard form of the generator matrix of the dual code $\C^\perp$ of a $\M$-additive code $\C$. We begin with defining a new inner product for the elements ${\bf u},{\bf v}\in\Z_{2}^{\alpha}\times\Z_{4}^{\beta}\times\Z_{8}^{\theta}$ as
\[
{\bf u}\cdot {\bf v} =4\left( \sum_{i=1}^{\alpha}u_{i}v_{i}\right)+2\left(\sum_{j=\alpha+1}^{\alpha+\beta}u_{j}v_{j}\right)+\sum_{k=\alpha+\beta+1}^{\alpha+\beta+\theta}u_{k}v_{k}.
\]
Further the dual code $\C^\perp$ can be defined in the usual way with respect to this inner product. 
\[
\C^{\perp }=\left\{{\bf v}\in\Z_{2}^{\alpha}\times\Z_{4}^{\beta}\times\Z_{8}^{\theta}\mid {\bf u}\cdot {\bf v} =0\text{ }for~all~{\bf u}\in\C\right\}.
\]
It is very easy to show that $\C^\perp$ is also a $\M$-additive code.

\begin{theo}
If $\C$ is a $\M$-additive code with the generator matrix in (\ref{G}),then
\begin{eqnarray}\label{H}
H =
 \left(\begin{array}{cc|ccc|cccc}
 -\bar{A_{01}}^{t} & I_{\alpha-k_{0}} & -2S_{1}^t & 0 & 0 &           \\
-T_{1}^{t} & 0 & -B_{02}^{t}-B_{12}^{t}B_{01}^{t} & B_{12}^t & I_{\beta-k_{1}-k_{2}}&      \\
0 & 0 & -2B_{01}^{t} & 2I_{k_2} & 0 & P     \\
-T_{2}^{t}& 0& -T_{4}^{t}+T_{3}^{t}A_{23}^{t}+T_{5}^{t}B_{01}^{t} & -T_{5}^{t} & 0 &    \\
0 & 0& -2T_{3}^{t} & 0 & 0 &   \\
0 & 0& 0 & 0 & 0 &
 \end{array}\right)
\end{eqnarray}
is the generator matrix of the dual code $\C^\perp$ (the parity-check matrix of  $\C$) where

$$P=\left(\begin{array}{cccc}
4\bar{S_{2}}^{t}-2\bar{S_{3}}^{t}A_{01}^{t}& -2\bar{S_{3}}^{t}& 0& 0 \\
-2S_{01}^{t}B_{12}^{t}-2S_{02}^{t}+2S_{12}^{t}A_{01}^{t}& -2S_{12}^{t}& 0& 0\\
-4S_{01}^{t}& 0& 0& 0\\
-A_{03}^{t}+A_{13}^{t}A_{01}^{t}+A_{23}^{t}A_{02}^{t}-A_{23}^{t}A_{12}^{t}A_{01}^{t}+2S_{01}^{t}T_{5}^{t}& -A_{13}^{t}+A_{23}^{t}A_{12}^{t}&-A_{23}^{t}&I_{\theta-k_{3}-k_{4}-k_{5}} \\
-2A_{02}^{t}+2A_{12}^{t}A_{01}^{t}& -2A_{12}^{t}&2I_{k_{5}}&0 \\
-4A_{01}^t& 4I_{k_{4}} &0&0
\end{array}\right).$$
\end{theo}
\begin{proof}
 We can easily check that $G_{S}\cdot H^{t}=0$. Therefore, every row of $H$ is orthogonal to the rows of $G_{S}$. Further, $|\C|=2^{k_{0}}\cdot4^{k_{1}}\cdot2^{k_{2}}\cdot8^{k_{3}}\cdot4^{k_{4}}\cdot2^{k_{5}}$ and $|\C^{\perp}|=2^{\alpha-k_{0}}\cdot4^{\beta-k_{1}-k_{2}}\cdot2^{k_{2}}\cdot8^{\theta-k_{3}-k_{4}-k_{5}}\cdot4^{k_{5}}\cdot2^{k_{4}}$. Hence, $|\C||\C^{\perp}|=\alpha+2\beta+3\theta$ and as a result, $H$ generates all of the code $\C^\perp$.
\end{proof}

\begin{coro}\label{TypeDual}
Let $\C$ be a $\M$-additive code of type $\type$ with standard form of the generator matrix (\ref{G}). Then, the dual  code $\C^\perp$ is of type $\left(\alpha,\beta,\theta;\alpha-k_{0};\beta-k_{1}-k_{2},k_{2};\theta-k_{3}-k_{4}-k_{5},k_{5},k_{4}\right)$.
\end{coro}

\begin{ex}
Let $\C$ be the $\M$-additive code with the generator matrix in (\ref{G_s}). Then using the above theorem, we can write the parity-check matrix of $\C$ as
$$H_S=\left(\begin{array}{cc|ccc|cccc}
1 & 1 & 1 & 0 & 1 & 6 & 0 & 0 & 0 \\
0 & 0 & 0 & 2 & 0 & 4 & 0 & 0 & 0 \\ [6pt] \hline
0 & 1 & 1 & 3 & 0 & 3 & 5 & 1 & 0 \\
1 & 1 & 2 & 0 & 0 & 3 & 5 & 0 & 1 \\
0 & 0 & 0 & 0 & 0 & 4 & 4 & 0 & 0\\
\end{array}\right).
$$
Therefore $\C^\perp$ is of type $\left(2,3,4;0;1,1;2,0,1\right)$ and has $2^{10}=1024$.
\end{ex}

\section{$\M$-cyclic codes}\label{sec:def}

Cyclic codes are significant family of linear codes because of their natural encoding and decoding algorithm. Moreover, since they can be described as ideals in some polynomial rings, they have a rich algebraic structure. In this section we introduce $\M$-cyclic codes and study their properties. We give their generator polynomials and also minimal generating sets.

\begin{defi}
Let $\C$ be a $\M$-additive code of length $n=\alpha+\beta+\theta$. $\C$ is called $\M$-cyclic if for any codeword \newline $c=\left(u_{0},u_{1},\dots,u_{\alpha-1},v_{0},v_{1},\dots,v_{\beta-1},w_{0},w_{1},\dots,w_{\theta-1}\right)\in \C,$ its cyclic shift
\[
T(c)=\left(u_{\alpha-1},u_{0},\dots,u_{\alpha-2},v_{\beta-1},v_{0},\dots,v_{\beta-2},w_{\theta-1},w_{0},\dots,w_{\theta-2}\right)
\]
is also in $\C$.
\end{defi}

\begin{lemm}
If $\C$ is a $\M$-cyclic code then the dual code $\C^{\perp}$ is also a cyclic $\M$ code.

\end{lemm}

\begin{proof}
Let $\C$ be a $\M$-cyclic code and $u=\left(a_{0},\dots,a_{\alpha-1},b_{0},\dots,b_{\beta-1},d_{0},\dots,d_{\theta-1}\right)\in \C^{\perp}$. We will show that $T(u)\in \C^{\perp}$. Since $u\in \C^{\perp}$, for $v=\left(e_{0},\dots,e_{\alpha-1},g_{0},\dots,g_{\beta-1},h_{0},\dots,h_{\theta-1}\right)$ we have
\begin{eqnarray*}
u\cdot v=4\left(a_{0}e_{0}+\cdots+a_{\alpha-1}e_{\alpha-1}\right)+2\left(b_{0}g_{0}+\cdots+b_{\beta-1}g_{\beta-1}\right)+d_{0}h_{0}+\cdots+d_{\theta-1}h_{\theta-1}=0~mod ~8.
\end{eqnarray*}
Now, let $m=lcm(\alpha,\beta,\theta)$, then $T^{m}(v)=v, v\in \Z_{2}^{\alpha}\times\Z_{4}^{\beta}\times\Z_{8}^{\theta}$. Let
$T^{m-1}(v)=\left(e_{1},\dots,e_{\alpha-1},e_{0},g_{1},\dots g_{0},h_{1},\dots h_{0}\right)=w$. Since $\C$ is cyclic $w\in \C$. Therefore,
\begin{eqnarray*}
0=w\cdot u&=&4\left(e_{1}a_{0}+\cdots+e_{0}a_{\alpha-1}\right)+2\left(g_{1}b_{0}+\cdots+g_{0}b_{\beta-1}\right)+h_{1}d_{0}+\cdots+h_{0}d_{\theta-1}\\
          &=&4\left(e_{0}a_{\alpha-1}+\cdots+e_{1}a_{0}\right)+2\left(g_{0}b_{\beta-1}+\cdots+g_{1}b_{0}\right)+d_{\theta-1}+\cdots+h_{1}d_{0}h_{0}\\
          &=&v\cdot T(u).
\end{eqnarray*}

Hence, $T(u)\in \C^{\perp}$ and $\C^{\perp}$ is also cyclic.
\end{proof}

\begin{defi}

We denote the module $\mathbb{Z}_2[x]/\langle x^{\alpha}-1\rangle\times\Z_{4} [x]/\langle x^{\beta}-1\rangle\times \Z_{8}[x]/\langle x^{\theta}-1\rangle$ by $\R$. Furthermore, any element
$$c=\left(u_{0},u_{1},\dots,u_{\alpha-1},v_{0},v_{1},\dots,v_{\beta-1},w_{0},w_{1},\dots,w_{\theta-1}\right)\in \C$$
can be identified with a module element consisting of three polynomials
\begin{eqnarray*}
c(x)&=&\left(u_{0}+u_{1}x+\cdots+u_{\alpha-1}x^{\alpha-1},v_{0}+v_{1}x+\cdots+v_{\beta-1}x^{\beta-1},w_{0}+w_{1}x+\cdots+w_{\theta-1}x^{\theta-1}\right)\\
&=&\left(u(x),v(x),w(x)\right)
\end{eqnarray*}
in $\R$. This is a one-one correspondence between the elements of $\Z_{2}^{\alpha}\times\Z_{4}^{\beta}\times\Z_{8}^{\theta}$ and $\R$. Moreover we can define a scalar multiplication $\ast$, for $d(x)\in \Z_{8}[x]$ and $\left(f(x),g(x),h(x)\in \R\right)$  as follows.
\[
d(x)\ast\left(f(x),g(x),h(x)\right)=\left(d(x)f(x)~mod~2,d(x)g(x)~mod~4,d(x)h(x)\right)
\]
\end{defi}

\begin{lemm}
The scalar multiplication $\ast$ defined above is well-defined and $\R$ is a $\Z_{8}[x]$-module with respect to this multiplication.
\end{lemm}

\subsection{The Structure of $\M$-Cyclic Codes}

In this subsection, we study submodules of $\M$. We describe the generators of such submodules and give their spanning sets. We always assume that $\alpha,~\beta$ and $\theta$ are all odd integers.
We have mentioned about that $\C$ is a $\Z_{8}[x]$ submodule of $\R$. Hence we define the following map.
\begin{eqnarray*}
\Psi:\C\longrightarrow \Z_{8}[x]/\langle x^{\theta}-1\rangle\\
\left(f(x),g(x),h(x)\right)\longrightarrow h(x)
\end{eqnarray*}

It can be easily seen that $\Psi$ is a module homomorphism. Furthermore, the image $\Psi(\C)$ of $\C$ is $\Z_{8}[x]$-submodule of $\Z_{8}[x]/ \langle x^{\theta}-1\rangle$ and $ker(\Psi)$ is a submodule of $\C$. Since we can look at to the image $\Psi(\C)$ as an ideal of the ring $\Z_{8}[x]/ \langle x^{\theta}-1\rangle$ and $\theta$ is an odd integer, from \cite{Calder} we can write
\[
\Psi(\C)=\langle p(x)+2q(x)+4r(x) \rangle~ \text{with}~r(x)|q(x)|p(x)|x^{\theta}-1~mod~8.
\]
Also the kernel of $\Psi$ is
\[
ker(\Psi)=\{\left(f(x),g(x),0\right)\in \C:f(x)\in \Z_{2}[x]/\langle x^{\alpha}-1 \rangle,g(x)\in \Z_{4}[x]/\langle x^{\beta}-1 \rangle \}.
\]

Now, define the set

\[
I=\{\left(f(x),g(x)\right)\in \Z_{2}[x]/\langle x^{\alpha}-1\rangle \times \Z_{4}[x]/\langle x^{\beta}-1\rangle :\left(f(x),g(x),0\right)\in ker(\Psi) \}.
\]
 It is obvious that $I$ is a $\Z_{2}\Z_{4}$-additive cyclic code in $\Z_{2}[x]/\langle x^{\alpha}-1\rangle \times \Z_{4}[x]/\langle x^{\beta}-1\rangle$. Therefore, from \cite{Abu} $I$ can be written in the following form.

 \[
I=\langle\left(f(x),0\right),\left(l_{1}(x),g(x)\right)\rangle
\]

where $f(x)|x^{\alpha}-1~mod~2$ and $g(x)=g_{1}(x)+2a_{1}(x)$ is a polynomial in $\Z_{4}[x]/\langle x^{\beta}-1\rangle$ with $a_{1}(x)|g_{1}(x)|x^{\beta}-1~mod~4$ and $l_{1}(x)$ is a binary polynomial of degree less than $f(x)$ such that $f(x)|\frac{x^{\beta}-1}{a_{1}(x)}l_{1}(x)$. Let $\left(c_{1}(x),c_{2}(x),0\right) \in ker(\Psi)$ then we have $\left(c_{1}(x),c_{2}(x)\right) \in I=\langle\left(f(x),0\right),\left(l_{1}(x),g(x)\right)\rangle$. So, for polynomials $m_{1}(x)\in \Z_{2}[x]/\langle x^{\alpha}-1 \rangle$ and $m_{2}(x)\in \Z_{4}[x]/\langle x^{\beta}-1 \rangle$,
\[
\left(c_{1}(x),c_{2}(x)\right)=m_{1}(x)*\left(f(x),0\right)+m_{2}(x)*\left(l_{1}(x),g_{1}(x)+2a_{1}(x) \right).
\]

Hence we can conclude that $ker(\Psi)$ is a submodule of $\C$ generated by

\[
ker(\Psi)=\langle \left(f(x),0,0\right), \left(l_{1}(x),g_{1}(x)+2a_{1}(x),0\right)\rangle.
\]

Finally, by the First Isomorphism Theorem we have

\[
\C/ker(\Psi)\cong \langle p(x)+2q(x)+4r(x)\rangle.
\]

Now, let $\left(l_{2}(x),g_{2}(x),p(x)+2q(x)+4r(x)\right)\in \C$ such that

\[
\Psi\left(l_{2}(x),g_{2}(x),p(x)+2q(x)+4r(x)\right)=\langle p(x)+2q(x)+4r(x)\rangle.
\]

Consequently, any $\M$-cyclic code $\C$ can be generated as a $\Z_{8}[x]$-submodule of $\R$ with the following form.
\[
\C=\langle \left(f(x),0,0\right),\left(l_{1}(x),g_{1}(x)+2a_{1}(x),0\right),\left(l_{2}(x),g_{2}(x),p(x)+2q(x)+4r(x)\right)\rangle.
\]

\begin{lemm}
Let $\C$ be a $\M$-cyclic code generated by
$$\left(\left(f(x),0,0\right),\left(l_{1}(x),g_{1}(x)+2a_{1}(x),0\right),\left(l_{2}(x),g_{2}(x),p(x)+2q(x)+4r(x)\right)\right).$$
Then $\deg(l_{1}(x))<\deg(f(x)),\deg(l_{2}(x))<\deg(f(x))$ and $\deg(g_{2}(x))<\deg(g_{1}(x))$.
\end{lemm}

\begin{proof}
Assume that $\deg(l_{1}(x))\geq\deg(f(x))$ with $\deg(l_{1}(x))-\deg(f(x))=i$. Now, let $\D$ be a $\M$-cyclic code with the generators of the form
\begin{eqnarray*}
&&\left(\left(f(x),0,0\right),\left(l_{1}(x)+x^{i}f(x),g_{1}(x)+2a_{1}(x),0\right)\right)\\
&&=\left(\left(f(x),0,0\right),\left(l_{1}(x),g_{1}(x)+2a_{1}(x),0\right)\right)+x^{i}*\left(f(x),0,0\right).
\end{eqnarray*}
Then, it is obvious that $\D\subseteq \C$. On the other hand,
\begin{eqnarray*}
\left(\left(l_{1}(x),g_{1}(x)+2a_{1}(x),0\right)\right)=\left(\left(l_{1}(x)+x^{i}f(x),g_{1}(x)+2a_{1}(x),0\right)\right)-x^{i}*\left(f(x),0,0\right).
\end{eqnarray*}
Therefore we have $\C\subseteq\D$ and hence $\C=\D$. The similar method can be used to prove $\deg(l_{2}(x))<\deg(f(x))$ and $\deg(g_{2}(x))<\deg(g_{1}(x))$.

\end{proof}

\begin{lemm}
If $\C=\langle\left(f(x),0,0\right),\left(l_{1}(x),g_{1}(x)+2a_{1}(x),0\right),\left(l_{2}(x),g_{2}(x),p(x)+2q(x)+4r(x)\right)\rangle$ is a $\M$-cyclic then we may assume that
\begin{enumerate}[label=(\roman*)]

\item
$f(x)|\frac{x^{\beta}-1}{a_{1}(x)}l_{1}(x)~mod~2$,

\item
$\left(g_{1}(x)+2a_{1}(x)\right)|\frac{x^{\theta}-1}{r(x)}g_{2}(x)~mod~4$,

\item

$f(x)|k(x)l_{1}(x)+\frac{x^{\theta}-1}{r(x)}l_{2}(x)~mod~2$ where $k(x)\left(g_{1}(x)+2a_{1}(x)\right)=\frac{x^{\theta}-1}{r(x)}g_{2}(x)$.

\end{enumerate}
\end{lemm}

\begin{proof}
\begin{enumerate}[label=(\roman*)]

\item
It can be easily seen from Lemma 10 in \cite{Abu}.

\item
Let us consider
\begin{eqnarray*}
\frac{x^{\theta}-1}{r(x)}*\left(l_{2}(x),g_{2}(x),p(x)+2q(x)+4r(x)\right)=\left(\frac{x^{\theta}-1}{r(x)}l_{2}(x),\frac{x^{\theta}-1}{r(x)}g_{2}(x),0\right).
\end{eqnarray*}
Since
\begin{eqnarray*}
\Psi\left(\frac{x^{\theta}-1}{r(x)}l_{2}(x),\frac{x^{\theta}-1}{r(x)}g_{2}(x),0\right)=0,
\end{eqnarray*}
$\left(\frac{x^{\theta}-1}{r(x)}l_{2}(x),\frac{x^{\theta}-1}{r(x)}g_{2}(x),0\right)\in ker(\Psi)\subseteq \C$ and therefore $\left(g_{1}(x)+2a_{1}(x)\right)|\frac{x^{\theta}-1}{r(x)}g_{2}(x)~mod~4$.

\item

Since $\left(g_{1}(x)+2a_{1}(x)\right)|\frac{x^{\theta}-1}{r(x)}g_{2}(x)~mod~4$ then we can write
\begin{eqnarray*}
\left(g_{1}(x)+2a_{1}(x)\right)k(x)=\frac{x^{\theta}-1}{r(x)}g_{2}(x),~k(x)\in \Z_{4}[x].
\end{eqnarray*}
Now, consider
\begin{eqnarray*}
k(x)*\left(l_{1}(x),g_{1}(x)+2a_{1}(x),0\right)=\left(k(x)l_{1}(x),k(x)\left(g_{1}(x)+2a_{1}(x)\right),0\right) \in \C.
\end{eqnarray*}
On the other hand, $\left(\frac{x^{\theta}-1}{r(x)}l_{2}(x),\frac{x^{\theta}-1}{r(x)}g_{2}(x),0\right)\in \C$. Hence,
\begin{eqnarray*}
&&\left(k(x)l_{1}(x),k(x)\left(g_{1}(x)+2a_{1}(x)\right),0\right)-\left(\frac{x^{\theta}-1}{r(x)}l_{2}(x),\frac{x^{\theta}-1}{r(x)}g_{2}(x),0\right)\\
&&=\left(k(x)l_{1}(x)+\frac{x^{\theta}-1}{r(x)}l_{2}(x),0,0\right)\in ker(\Psi)\subseteq \C.
\end{eqnarray*}

So, $f(x)|k(x)l_{1}(x)+\frac{x^{\theta}-1}{r(x)}l_{2}(x)~mod~2$.

\end{enumerate}
\end{proof}

We summarize all this discussion with the following theorem.

\begin{theo}\label{main}
Let $\C$ be a $\M$-cyclic code. Then we can identify $\C$ as

\begin{enumerate}[label=(\roman*)]
\item
$\C=\langle f(x),0,0 \rangle$ with $f(x)|(x^{\alpha}-1)~mod~2$.

\item
$\C=\langle \left(l_{1}(x),g_{1}(x)+2a_{1}(x),0 \right)\rangle$ with $a_{1}(x)|g_{1}(x)|(x^{\beta}-1)~mod~4$ and $l_{1}(x)$ is a binary polynomial of degree less than $f(x)$ such that $f(x)|\frac{x^{\beta}-1}{a_{1}(x)}l_{1}(x)~mod~2$.

\item
$\C=\langle \left(l_{2}(x),g_{2}(x),p(x)+2q(x)+4r(x)\right)\rangle$, where $r(x)|q(x)|p(x)|(x^{\theta}-1)~mod~8$, and $g_{2}(x))$ is a polynomial in $\Z_{4}[x]$ of degree less than $g_{1}(x)$ such that $\left(g_{1}(x)+2a_{1}(x)\right)|\frac{x^{\theta}-1}{r(x)}g_{2}(x)~mod~4$. And also $l_{2}(x)$ is a binary polynomial of degree less than $f(x)$ with $f(x)|k(x)l_{1}(x)+\frac{x^{\theta}-1}{r(x)}l_{2}(x)~mod~2$ where $k(x)\left(g_{1}(x)+2a_{1}(x)\right)=\frac{x^{\theta}-1}{r(x)}g_{2}(x)$.

\item

$\C=\langle \left(f(x),0,0\right),\left(l_{1}(x),g_{1}(x)+2a_{1}(x),0\right),\left(l_{2}(x),g_{2}(x),p(x)+2q(x)+4r(x)\right)\rangle$ where the generator polynomials are defined as above.

\end{enumerate}

\end{theo}

\begin{theo}\label{sset}
Let $$\C=\langle \left(f(x),0,0\right),\left(l_{1}(x),g_{1}(x)+2a_{1}(x),0\right),\left(l_{2}(x),g_{2}(x),p(x)+2q(x)+4r(x)\right)\rangle$$ be a $\M$-cyclic code with the generator polynomials defined above. Define the sets
\begin{eqnarray*}
S_1&=&\bigcup^{\deg(h_{f}(x))-1}_{i=0} \Big\{  x^i\ast(f(x),0,0)\Big\}, \\
S_2&=&\bigcup^{\deg(h_{1}(x))-1}_{i=0}\Big\{x^i\ast(l_{1}(x),g_{1}(x)+2a_{1}(x),0)\Big\}, \\
S_3&=& \bigcup^{\deg(h_p(x))-1}_{i=0}\Big\{x^i\ast(l_{2}(x),g_{2}(x),p(x)+2q(x)+4r(x))\Big\},  \\
S_4&=& \bigcup^{\deg(\hat{q}(x))-1}_{i=0}\Big\{x^i\ast(h_{p}(x)l_{2}(x),h_{p}(x)g_{2}(x),2h_{p}(x)\left(q(x)+2r(x)\right))\Big\},\\
S_5&=& \bigcup^{\deg(\hat{r}(x))-1}_{i=0}\Big\{x^i\ast(h_{q}(x)l_{2}(x),h_{q}(x)g_{2}(x),4h_{q}(x)r(x))\Big\},\text{  and }\\
S_6&=& \bigcup^{\deg(b_{1}(x))-1}_{i=0}\Big\{x^i\ast(h_{1}(x)l_{1}(x),2a_{1}(x)h_{1}(x),0)\Big\}
\end{eqnarray*}
where $f(x)h_{f}(x)=x^{\alpha}-1$, $g_{1}(x)h_{1}(x)=x^{\beta}-1,~g_{1}(x)=a_{1}(x)b_{1}(x)$, $\hat{q}(x)=\frac{p(x)}{q(x)},~\hat{r}(x)=\frac{q(x)}{r(x)}$ and $p(x)h_{p}(x)=q(x)h_{q}(x)=x^{\theta}-1$.
Then, $S=S_1\cup S_2\cup S_3\cup S_4\cup S_5\cup S_6$ forms a minimal generating set for the $\M$-cyclic code $\C$ and $\C$ has $2^{\deg(h_{f})}\cdot4^{\deg(h_{1})}\cdot8^{\deg(h_{p})}\cdot4^{\deg(\hat{q})}\cdot2^{\deg(\hat{r})}\cdot2^{\deg(b_{1})}$ codewords.
\end{theo}

\begin{proof}
Let us take any codeword $c(x)$ in $\C$. Therefore, for polynomials $e_{1}(x),e_{2}(x),e_{3}(x)\in \Z_{8}[x]$, $c(x)$ can be written as,
\[
c(x)=e_{1}(x)\ast\left(f(x),0,0\right)+e_{2}(x)\ast\left(l_{1}(x),g_{1}(x)+2a_{1}(x),0\right)+e_{3}(x)\ast\left(l_{2}(x),g_{2}(x),p(x)+2q(x)+4r(x)\right).
\]

If $\deg e_{1}(x)\leq \deg(h_{f}(x))-1$ then $e_{1}(x)\ast\left(f(x),0,0\right)\in Span(S_{1})$. Otherwise, with the help of division algorithm we get

\[
e_{1}(x)=h_{f}(x)q_{1}(x)+r_{1}(x)
\]
with $r_{1}(x)=0$ or $\deg(r_{1}(x))\leq \deg(h_{f}(x))-1$. Hence,
\begin{eqnarray*}
e_{1}(x)\ast \left(f(x),0,0\right)&=&\left(h_{f}(x)q_{1}(x)+r_{1}(x)\right)\ast \left(f(x),0,0\right)\\
&=&r_{1}(x)\ast\left(f(x),0,0\right)\in Span(S_{1}).
\end{eqnarray*}

Now, if $\deg(e_{2}(x))\leq \deg(h_{1}(x))-1$ then $e_{2}(x)\ast\left(l_{1}(x),g_{1}(x)+2a_{1}(x),0\right)\in Span(S_{2})$. Otherwise by using the division algorithm again, we have polynomials $q_{2}(x)$ and $r_{2}(x)$ such that
\[
e_{2}(x)=q_{2}(x)h_{1}(x)+r_{2}(x)
\]
where $r_{2}(x)=0$ or $\deg(r_{2}(x))\leq \deg(h_{1}(x))-1$.
So,
\begin{eqnarray*}
	e_{2}(x)\ast\left(l_{1}(x),g_{1}(x)+2a_{1}(x),0\right)&=&\left(q_{2}(x)h_{1}(x)+r_{2}(x)\right)\ast\left(l_{1}(x),g_{1}(x)+2a_{1}(x),0\right)\\
	&=&q_{2}(x)\ast\left(h_{1}(x)l_{1}(x),2a_{1}(x)h_{1}(x),0\right)+r_{2}(x)\ast\left(l_{1}(x),g_{1}(x)+2a_{1}(x),0\right).
\end{eqnarray*}

Here, $r_{2}(x)\ast\left(l_{1}(x),g_{1}(x)+2a_{1}(x),0\right)\in Span(S_{2})$, so we only consider $q_{2}(x)\ast\left(h_{1}(x)l_{1}(x),2a_{1}(x)h_{1}(x),0\right)$. If $\deg(q_{2}(x))\leq \deg(b_{1}(x))-1$ then $q_{2}(x)\ast\left(h_{1}(x)l_{1}(x),2a_{1}(x)h_{1}(x),0\right)\in Span(S_{6})$, otherwise

\[
q_{2}(x)=q_{3}(x)b_{1}(x)+r_{3}(x),~r_{3}(x)=0~\text{or}~\deg(r_{3}(x))\leq \deg(b_{1}(x))-1.
\]

Hence,
\begin{eqnarray*}
	q_{2}(x)\ast\left(h_{1}(x)l_{1}(x),2a_{1}(x)h_{1}(x),0\right)&=&\left(q_{3}(x)b_{1}(x)+r_{3}(x)\right)\ast\left(h_{1}(x)l_{1}(x),2a_{1}(x)h_{1}(x),0\right)\\
	&=&q_{3}(x)\ast\left(h_{1}(x)l_{1}(x)b_{1}(x),0,0\right)+r_{3}(x)\ast\left(h_{1}(x)l_{1}(x),2a_{1}(x)h_{1}(x),0\right).
\end{eqnarray*}

It is clear that $r_{3}(x)\ast\left(h_{1}(x)l_{1}(x),2a_{1}(x)h_{1}(x),0\right)\in Span(S_{6})$ and also since $f(x)|\frac{x^{\beta}-1}{a_{1}(x)}l_{1}(x)$,
\[
q_{3}(x)\ast\left(h_{1}(x)l_{1}(x)b_{1}(x),0,0\right)=q_{3}(x)\ast\left(\frac{x^{\beta}-1}{a_{1}(x)}l_{1}(x),0,0\right)\in Span(S_{1}).
\]

Now, let us consider $e_{3}(x)\ast\left(l_{2}(x),g_{2}(x),p(x)+2q(x)+4r(x)\right)$.

If $\deg(e_{3}(x))\leq \deg(h_{p}(x))-1$ then we are done, i.e., $e_{3}(x)\ast\left(l_{2}(x),g_{2}(x),p(x)+2q(x)+4r(x)\right)\in Span(S_{3})$. Otherwise
\[
e_{3}(x)=q_{4}(x)h_{p}(x)+r_{4}(x)
\]
where $r_{4}(x)=0$ or $\deg(r_{4}(x))\leq \deg(h_{p}(x))-1$.
Using the same way above, we have
\begin{eqnarray*}
e_{3}(x)\ast\left(l_{2}(x),g_{2}(x),p(x)+2q(x)+4r(x)\right)&=&\left(q_{4}(x)h_{p}(x)+r_{4}(x)\right)\ast\left(l_{2}(x),g_{2}(x),p(x)+2q(x)+4r(x)\right)\\
&=&q_{4}(x)\ast\left(h_{p}(x)l_{2}(x),h_{p}(x)g_{2}(x),2h_{p}(x)q(x)+4h_{p}(x)r(x)\right)\\
&&+r_{4}(x)\ast\left(l_{2}(x),g_{2}(x),p(x)+2q(x)+4r(x)\right).
\end{eqnarray*}

Here, $r_{4}(x)\ast\left(l_{2}(x),g_{2}(x),p(x)+2q(x)+4r(x)\right)$ is included by $S_{3}$, so it only remains $q_{4}(x)\ast\left(h_{p}(x)l_{2}(x),h_{p}(x)g_{2}(x),2h_{p}(x)q(x)+4h_{p}(x)r(x)\right)$. If $\deg(q_{4}(x))\leq \deg( \hat{q}(x))-1$ then $q_{4}(x)\ast\left(h_{p}(x)l_{2}(x),h_{p}(x)g_{2}(x),2h_{p}(x)q(x)+4h_{p}(x)r(x)\right)\in Span(S_{4})$. Otherwise

\[
q_{4}(x)=q_{5}(x)\hat{q}(x)+r_{5}(x),~\text{where}~r_{5}(x)=0~\text{or}~\deg(r_{5}(x))\leq \deg(\hat{q}(x))-1.
\]

Therefore,
\begin{eqnarray*}
&&q_{4}(x)\ast\left(h_{p}(x)l_{2}(x),h_{p}(x)g_{2}(x),2h_{p}(x)q(x)+4h_{p}(x)r(x)\right)=q_{5}(x)\ast\left(h_{q}(x)l_{2}(x),h_{q}(x)g_{2}(x),4h_{q}(x)r(x)\right)\\
&&+r_{5}(x)\ast\left(h_{p}(x)l_{2}(x),h_{p}(x)g_{2}(x),2h_{p}(x)q(x)+4h_{p}(x)r(x)\right).
\end{eqnarray*}

$r_{5}(x)\ast\left(h_{p}(x)l_{2}(x),h_{p}(x)g_{2}(x),2h_{p}(x)q(x)+4h_{p}(x)r(x)\right)$ is in $Span(S_{4})$. For $q_{5}(x)\ast\left(h_{q}(x)l_{2}(x),h_{q}(x)g_{2}(x),4h_{q}(x)r(x)\right)$, if $\deg(q_{5}(x))\leq \deg(\hat{r}(x))-1$ then we are done. Otherwise
\[
q_{5}(x)=q_{6}(x)\hat{r}(x)+r_{6}(x),~r_{6}(x)=0~\text{or}~\deg(r_{6}(x))\leq \deg(\hat{r}(x))-1.
\]
Using the division algorithm for the last time, we get
\begin{eqnarray*}
&&\left(q_{6}(x)\hat{r}(x)+r_{6}(x)\right)\ast\left(h_{q}(x)l_{2}(x),h_{q}(x)g_{2}(x),4h_{q}(x)r(x)\right)\\
&=&q_{6}(x)\ast\left(\frac{x^{\theta}-1}{r(x)}l_{2}(x),\frac{x^{\theta}-1}{r(x)}g_{2}(x),0\right)\\
&&+r_{6}(x)\ast\left(h_{q}(x)l_{2}(x),h_{q}(x)g_{2}(x),4h_{q}(x)r(x)\right).
\end{eqnarray*}

We have $r_{6}(x)\ast\left(h_{q}(x)l_{2}(x),h_{q}(x)g_{2}(x),4h_{q}(x)r(x)\right)\in S_{5}$. Furthermore, from Theorem \ref{main} we know $k(x)\left(g_{1}(x)+2a_{1}(x)\right)=\frac{x^{\theta}-1}{r(x)}g_{2}(x)$ and
\begin{eqnarray*}
f(x)|k(x)l_{1}(x)+\frac{x^{\theta}-1}{r(x)}l_{2}(x)&\Rightarrow& f(x)\mu(x)=k(x)l_{1}(x)+\frac{x^{\theta}-1}{r(x)}l_{2}(x)\\
&\Rightarrow& \frac{x^{\theta}-1}{r(x)}l_{2}(x)= f(x)\mu(x)-k(x)l_{1}(x).
\end{eqnarray*}

Hence,
\begin{eqnarray*}
&&q_{6}(x)\ast\left(\frac{x^{\theta}-1}{r(x)}l_{2}(x),\frac{x^{\theta}-1}{r(x)}g_{2}(x),0\right)\\
&=& q_{6}(x)\ast\left(f(x)\mu(x)-k(x)l_{1}(x),k(x)\left(g_{1}(x)+2a_{1}(x)\right),0\right)\\
&=& q_{6}(x) \left[ \mu(x)\left(f(x),0,0\right)+ k(x)\left(l_{1}(x),g_{1}(x)+2a_{1}(x),0\right)\right]\in \left(S_{1}\cup S_{2}\right).
\end{eqnarray*}

Finally we have proved the theorem.
\end{proof}

\begin{ex}
Let $\C$ be $\M$-cyclic code in $\Z_{2}[x]/\langle x^{15}-1 \rangle\times\Z_{4}[x]/\langle x^7-1\rangle\times\Z_{8}[x]/\langle x^{7}-1 \rangle$ generated by
$$ \left(\left(f(x),0,0\right),\left(l_{1}(x),g_{1}(x)+2a_{1}(x),0\right),\left(l_{2}(x),g_{2}(x),p(x)+2q(x)+4r(x)\right)\right)$$
with
\begin{eqnarray*}
f(x)&=&1+x+x^3+x^5,~l_{1}(x)=l_{2}(x)=1+x^3+x^4,\\
g_{1}(x)&=&1+2 x+3 x^2+x^3+x^4,~a_{1}(x)=3+x+2 x^2+x^3,~g_{2}(x)=3+ x,\\
p(x)&=&1+5 x+7 x^2+2 x^3+x^4,~q(x)=7+x,~r(x)=1.
\end{eqnarray*}

Further we can calculate the following polynomials.
\begin{eqnarray*}
p(x)h_{p}(x)&=&x^{7}-1\Longrightarrow h_{p}(x)=7+5 x+6 x^2+x^3,\\
q(x)h_{q}(x)&=&x^{7}-1\Longrightarrow h_{q}(x)=1+x+x^2+x^3+x^4+x^5+x^6,\\
k(x)\left(g_{1}(x)+2a_{1}(x)\right)&=&\frac{x^{\theta}-1}{r(x)}g_{2}(x)\Longrightarrow k(x)=1+x,\\
g_{1}(x)h_{1}(x)&=&x^{7}-1\Longrightarrow h_{1}(x)=3+2 x+3 x^2+x^3.
\end{eqnarray*}
Hence using the generator sets in Theorem \ref{sset}, we can write the generator matrix for $\C$ as follows.
\[
\left(
\begin{array}{ccccccccccccccccccccccccccccc}
 1 & 1 & 0 & 1 & 0 & 1 & 0 & 0 & 0 & 0 & 0 & 0 & 0 & 0 & 0 & 0 & 0 & 0 & 0 & 0 & 0 & 0 & 0 & 0 & 0 & 0 & 0 & 0 & 0 \\
 0 & 1 & 1 & 0 & 1 & 0 & 1 & 0 & 0 & 0 & 0 & 0 & 0 & 0 & 0 & 0 & 0 & 0 & 0 & 0 & 0 & 0 & 0 & 0 & 0 & 0 & 0 & 0 & 0 \\
 0 & 0 & 1 & 1 & 0 & 1 & 0 & 1 & 0 & 0 & 0 & 0 & 0 & 0 & 0 & 0 & 0 & 0 & 0 & 0 & 0 & 0 & 0 & 0 & 0 & 0 & 0 & 0 & 0 \\
 0 & 0 & 0 & 1 & 1 & 0 & 1 & 0 & 1 & 0 & 0 & 0 & 0 & 0 & 0 & 0 & 0 & 0 & 0 & 0 & 0 & 0 & 0 & 0 & 0 & 0 & 0 & 0 & 0 \\
 0 & 0 & 0 & 0 & 1 & 1 & 0 & 1 & 0 & 1 & 0 & 0 & 0 & 0 & 0 & 0 & 0 & 0 & 0 & 0 & 0 & 0 & 0 & 0 & 0 & 0 & 0 & 0 & 0 \\
 0 & 0 & 0 & 0 & 0 & 1 & 1 & 0 & 1 & 0 & 1 & 0 & 0 & 0 & 0 & 0 & 0 & 0 & 0 & 0 & 0 & 0 & 0 & 0 & 0 & 0 & 0 & 0 & 0 \\
 0 & 0 & 0 & 0 & 0 & 0 & 1 & 1 & 0 & 1 & 0 & 1 & 0 & 0 & 0 & 0 & 0 & 0 & 0 & 0 & 0 & 0 & 0 & 0 & 0 & 0 & 0 & 0 & 0 \\
 0 & 0 & 0 & 0 & 0 & 0 & 0 & 1 & 1 & 0 & 1 & 0 & 1 & 0 & 0 & 0 & 0 & 0 & 0 & 0 & 0 & 0 & 0 & 0 & 0 & 0 & 0 & 0 & 0 \\
 0 & 0 & 0 & 0 & 0 & 0 & 0 & 0 & 1 & 1 & 0 & 1 & 0 & 1 & 0 & 0 & 0 & 0 & 0 & 0 & 0 & 0 & 0 & 0 & 0 & 0 & 0 & 0 & 0 \\
 0 & 0 & 0 & 0 & 0 & 0 & 0 & 0 & 0 & 1 & 1 & 0 & 1 & 0 & 1 & 0 & 0 & 0 & 0 & 0 & 0 & 0 & 0 & 0 & 0 & 0 & 0 & 0 & 0 \\
 1 & 0 & 0 & 1 & 1 & 0 & 0 & 0 & 0 & 0 & 0 & 0 & 0 & 0 & 0 & 3 & 0 & 3 & 3 & 1 & 0 & 0 & 0 & 0 & 0 & 0 & 0 & 0 & 0 \\
 0 & 1 & 0 & 0 & 1 & 1 & 0 & 0 & 0 & 0 & 0 & 0 & 0 & 0 & 0 & 0 & 3 & 0 & 3 & 3 & 1 & 0 & 0 & 0 & 0 & 0 & 0 & 0 & 0 \\
 0 & 0 & 1 & 0 & 0 & 1 & 1 & 0 & 0 & 0 & 0 & 0 & 0 & 0 & 0 & 0 & 0 & 3 & 0 & 3 & 3 & 1 & 0 & 0 & 0 & 0 & 0 & 0 & 0 \\
 1 & 0 & 0 & 1 & 1 & 0 & 0 & 0 & 0 & 0 & 0 & 0 & 0 & 0 & 0 & 3 & 1 & 0 & 0 & 0 & 0 & 0 & 3 & 7 & 7 & 2 & 1 & 0 & 0 \\
 0 & 1 & 0 & 0 & 1 & 1 & 0 & 0 & 0 & 0 & 0 & 0 & 0 & 0 & 0 & 0 & 3 & 1 & 0 & 0 & 0 & 0 & 0 & 3 & 7 & 7 & 2 & 1 & 0 \\
 0 & 0 & 1 & 0 & 0 & 1 & 1 & 0 & 0 & 0 & 0 & 0 & 0 & 0 & 0 & 0 & 0 & 3 & 1 & 0 & 0 & 0 & 0 & 0 & 3 & 7 & 7 & 2 & 1 \\
 1 & 0 & 1 & 0 & 1 & 1 & 0 & 1 & 0 & 0 & 0 & 0 & 0 & 0 & 0 & 2 & 2 & 2 & 2 & 2 & 2 & 2 & 0 & 0 & 0 & 0 & 0 & 0 & 0 \\
 1 & 1 & 0 & 0 & 0 & 1 & 1 & 1 & 0 & 0 & 0 & 0 & 0 & 0 & 0 & 1 & 2 & 3 & 1 & 1 & 0 & 0 & 6 & 0 & 6 & 6 & 2 & 0 & 0 \\
 0 & 1 & 1 & 0 & 0 & 0 & 1 & 1 & 1 & 0 & 0 & 0 & 0 & 0 & 0 & 0 & 1 & 2 & 3 & 1 & 1 & 0 & 0 & 6 & 0 & 6 & 6 & 2 & 0 \\
 0 & 0 & 1 & 1 & 0 & 0 & 0 & 1 & 1 & 1 & 0 & 0 & 0 & 0 & 0 & 0 & 0 & 1 & 2 & 3 & 1 & 1 & 0 & 0 & 6 & 0 & 6 & 6 & 2 \\
 1 & 1 & 1 & 0 & 1 & 1 & 1 & 0 & 0 & 0 & 1 & 0 & 0 & 0 & 0 & 0 & 0 & 0 & 0 & 0 & 0 & 0 & 4 & 4 & 4 & 4 & 4 & 4 & 4
\end{array}
\right)
\]
Furthermore, the Gray image $\Phi(\C)$ of $\C$ is a $[57,33,4]$ linear binary code.

\end{ex}

\begin{section}{Conclusion}

In this work we introduce $\M$-additive and cyclic codes. We give the standard forms of the generator matrices of both $\C$ and its dual $\C^{\perp}$. For the cyclic case, we determine generator polynomials and a spanning sets of a $\M$-cyclic code $\C$. We also give examples for both cases. Since this family of codes are very new, there are many things to explore, for example self-dual codes. 

\end{section}



\end{document}